\documentclass[12pt]{article}
\usepackage[margin=1in]{geometry}

\usepackage{style/qstyle}

\title{Sequential transmission at short times}

\author[1,2]{Archishna Bhattacharyya \thanks{abhat086@uottawa.ca}}
\affil[1]{Perimeter Institute for Theoretical Physics}
\affil[2]{Department of Mathematics and Statistics, University of Ottawa}
\date{}

\numberwithin{equation}{section}
\numberwithin{figure}{section}

\begin{document}

\maketitle

\vspace{0.2cm}

\begin{abstract}
     We show that it is possible to transmit and preserve information at short time scales over an n-fold composition of quantum channels $(\Xi^n)_{n \in \mbb{N}}$ modelled as a discrete quantum Markov semigroup, long enough to generate entanglement at some finite $n$. This is achieved by interspersing the action of noise with quantum error correction in succession. We show this by means of a non-trivial lower bound on the one-shot quantum capacity in the sequential setting as a function of $n$, in an attempt to model a linear quantum network and assess its capabilities to distribute entanglement. Intriguingly, the rate of transmission of such a network turns out to be a property of the spectrum of the channels composed in sequence, and the maximum possible error in transmission can be bounded as a function of the noise model only. As an application, we derive an exact error bound for the infinite dimensional pure-loss channel believed to be the dominant source of noise in networks precluding the distribution of entanglement. We exemplify our results by analysing the amplitude damping channel and its bosonic counterpart.
\end{abstract}

\vspace{1cm}

\tableofcontents

\newpage

\section{Introduction}

The possibility of a quantum internet \cite{kimble2008quantum} is intriguing, and as such warrants an investigation of the capabilities of quantum networks to distribute entanglement. The simplest quantum network may be envisioned as a sequence of quantum channels (i.e., linear completely positive trace-preserving maps), over which information may be transmitted. A major impediment to successful transmission of information over any such network is a noisy environment which may also be described by a quantum channel. Each individual quantum channel in such a sequence models a node $\Xi$, say.  Thus, one may think of a linear quantum network as a sequence of nodes interspersed with the action of a noisy channel in succession. In this context, assessing the capabilities of a network is tantamount to characterising the information-transmission capacities of the sequence of channels as a whole $\left(\Xi \circ \Xi \circ \cdots \circ \Xi = \Xi^n\right)_{n \in \mathbb{N}}$, ideally over a single use. This contrasts the well-studied setting of evaluating quantum channel capacities over several uses in \textit{parallel} $\Xi^{\otimes n}$, in quantum Shannon theory. In fact, this is precisely the problem of bounding the one-shot capacities of an n-fold composition of channels \textit{sequentially}, recently witnessing enormous progress. It was studied in the asymptotic limit by \cite{singh2024zero, guan2018decomposition} for zero-error capacities, by \cite{fawzi2024capacities} for finite error; and over a finite number of sequential uses with non-vanishing error probability by \cite{singh2024information}. Most recently, \cite{singh2025capacities} showed that for channels $\Phi$ with the property $\Phi^l = \Psi$ for some large $l$, interesting Shannon-theoretic properties like strong additivity, strong converse property, coincidence of quantum and private capacities were approximately and exactly satisfied respectively in the finite-error and zero-error cases in the sequential setting, under certain conditions. 

Yet, in all of these works \cite{singh2024zero, guan2018decomposition, fawzi2024capacities, singh2024information, singh2025capacities}, it was noted that finding a non trivial lower bound on the capacity at \textit{short times} or \textit{small n} for $(\Xi^n)_{n \in \mathbb{N}}$ is an open problem. Our work is the first to find a solution. In fact, this problem is what lies at the heart of characterising a noise-robust quantum network  and is indispensable for tasks like distributed quantum computing \cite{grover1997quantum, cirac1999distributed}, blind quantum computing \cite{broadbent2009universal}, secure communication \cite{Vidick_Wehner_2023} and many more that constitute elements of a quantum internet. 

We consider sending quantum data over an n-fold composition of channels $\left(\Xi^n\right)_{n \in \mbb{N}}$. Each $\Xi$ consists of an encoding $\mc{E}$ and decoding $\mc{D}$ operation with the intervening action of a noisy channel $\mc{N}$ as $\Xi = \mc{D} \circ \mc{N} \circ \mc{E}$ (see \cref{def: segment-channel}). Our primary interest is in quantifying a rate of transmission through $(\Xi^n)_{n \in \mbb{N}}$ that tells us how much information we can salvage in the presence of noise as a function of the time scale $n$ (or the number of channels composed); for some finite $n$. Thus, by finding a lower bound \textit{without any computational or physical assumptions} on the sequential one-shot quantum capacity of $(\Xi^n)_{n \in \mbb{N}}$ with explicit $n-$dependence, we hope to determine its efficacy as a model for linear quantum networks permitting finite error in transmission at any instant of time. 

Our approach complements the scenario studied in \cite{singh2024information,singh2025capacities} wherein the ability of a quantum memory, modelled as a \textit{discrete Quantum Markov Semigroup} (see \cref{def: dQMS}), to store information in a noisy environment is quantified by upper and lower bounds on the one-shot capacities of an n-fold composition of channels. This is because we consider the evolution at short time scales. The analysis in \cite{singh2024information}, relevant to our context, quantifies how rapidly such a composition of channels approaches its late-time behaviour described by its limit point — a completely mixing channel, $\Xi_{\infty}$ (see \cref{def: comp-mixed-ch}) mapping every input to the same output. We wish to explore whether it is at all possible to approach its late-time behaviour slowly enough, such that information can be transmitted in a noise-robust manner over $\Xi^n$ in finite time. This question remaining open is evident from the fact that in the simplest case of $\Xi$ being a qubit channel, the lower bound on the one-shot quantum capacity is trivial, i.e.,  either zero or unity arising from the possible block dimensions of the \textit{peripheral space} \cite[Theorem 3.1]{singh2024information} of the channel. Moreover, the lower bound obtained is independent of the number of channels composed, thus a priori bearing no direct relationship of the transmission with time. We consider the setting where \textit{error correction} is allowed between time steps in the evolution of $(\Xi^n)_{n \in \mbb{N}}$ contrary to \cite{singh2024information}.

We show that by means of simple quantum error correcting codes \cite{knill1996theory}, it is possible to transmit and preserve information at short times over an n-fold composition of channels $\left(\Xi^n\right)_{n \in \mathbb{N}}$ long enough to generate entanglement between distant points. This is necessary to accurately model a quantum network that is noise-robust. Another facet of this problem is to understand which codes facilitate better transmission in spite of some permissible error, and if it possible to bound the total error occurring in the noise model. This naturally posits a nuanced relationship between quantum error correction and the quality of transmission over $(\Xi^n)_{n \in \mbb{N}}$ which needs to be quantified in the one-shot sequential setting. We address this gap. In related directions, studies have focused on elaborate constructions of codes, among other primitives, and schemes leading to reliable transmission, see for example \cite{wo2023resource}. However, insights into the structure of simply a linear quantum network that validate its existence remain elusive. Moreover, it is of theoretical interest to non trivially quantify information transmission rates at finite time scales in the sequential setting of quantum Shannon theory. We advance this understanding.

\paragraph{Summary of results} We derive an analytic non trivial lower bound on the one-shot quantum capacity of a sequence of channels, $\Xi^n$ in \cref{thm: continuity-capacity} valid at small $n$, or short time scales. The bound is non-trivial in the sense that it is a function of the number of channels composed, hence the time scale when viewed as a discrete Quantum Markov Semigroup. In \cref{thm: rad-conv}, we quantify structural properties of $\Xi^n$ with analytic conditions under which it preserves or dissipates information over time, using the spectral analysis of channels considered in \cite{szehr2015spectral} and \cite{ruskai2002analysis}. Finally, we provide an analytic bound on the error that can occur during transmission as a function of the noise-model in \cref{thm: error-bound}. As an application, we quantify the error bound of \cref{thm: error-bound} for the infinite-dimensional pure loss channel that is a dominant source of noise in realistic quantum networks. Further, we exemplify our main results by analysing the amplitude damping channel and its bosonic counterpart for capacity and error bounds.

\paragraph{Outline} \cref{sec: prelim} reviews background material to follow the main results in \cref{sec: main}. \cref{sec: app} presents an application to physically relevant noise models of the infinite dimensional pure-loss channel, the amplitude damping channel and its higher-dimensional analogues. \cref{sec: conc} concludes with a summary and outlook.

\section{Preliminaries} \label{sec: prelim}

\paragraph{Notation} For $n\in\N$, write $[n]=\{1,2,\ldots,n\}$. We write $\log$ for the base-$2$ logarithm. We denote quantum systems or registers by capital Latin letters $A, B, C$ and associated Hilbert spaces by capital script letters $\mc{H}_A, \mc{H}_B, \mc{H}_C.$ We work with separable Hilbert spaces which may be infinite dimensional. A joint quantum system $AB$ has Hilbert space $\mc{H}_A \otimes \mc{H}_B \equiv \mc{H}_{AB}.$ The set of bounded linear operators on $\mc{H}$ is denoted by $B(\mc{H}).$ An operator $A \in B(\mc{H})$ is positive, denoted by $A \geq 0$, if $A = \sqrt{A^{\dagger}A}$, where $(\cdot)^{\dagger}$ represents the Hermitian conjugate. We denote by $\text{id}_R$ the identity map on $B(\mc{H}_R)$. We write $\Tr$ for the trace on $B(\mc{H})$. Depending on the context, we write $\Tr$ for the trace on trace class operators $T(\mc{H})$, for separable $\mc{H}$. On $B(\mc{H}_{AB})$, we write the partial trace $\Tr_{B}=\id\otimes\Tr$. For $\rho_{AB}\in B(\mc{H}_{AB})$, write $\rho_A=\Tr_B(\rho_{AB})$. We denote the $L_1$-norm by $\norm{\cdot}_1$ and the trace norm by $\frac{1}{2}\norm{\cdot}_1$. We denote the operator norm by $\norm{\cdot}$. The set of quantum states or density operators is denoted by $D(\mc{H}): \{\rho \in B(\mc{H}), \rho \geq 0, \Tr{\rho} = 1\}$. The identity operator is denoted by $I$. In this work, the set of $2 \times 2$ matrices with complex-valued entries will be frequently used to describe a qubit system and hence we denote such a set by $M_2.$ 

A quantum channel $\Phi: B(\mc{H}) \to B(\mc{H})$ is a linear, completely positive trace preserving (CPTP) map. The Kraus representation of a quantum channel $\Phi: B(\mc{H}) \to B(\mc{H})$ is $\sum \limits_{i = 0}^d \Phi(\rho) = A_i \rho A_i^{\dagger}$ where $A_i \in B(\mc{H})$ are the Kraus operators such that $\sum \limits_{i = 0 }^d A_i^{\dagger} A_i = I,$ and $d = \text{dim}(\mc{H})$. If $\mc{H}$ is separable, then $\sum \limits_{i = 0 }^{\infty} A_i^{\dagger} A_i$ converges to the identity in the strong operator topology, where $A_i \in T(\mc{H}).$ For $\Phi: B(\mc{H}_A) \to B(\mc{H}_B)$, by Stinespring's dilation theorem, there exists an isometry $V: A \to BE$ known as the Stinespring isometry such that $\Phi (\rho) = \Tr_E (V \rho V^{\dagger})$. The complementary channel $\Phi^c: B(\mc{H}_A) \to B(\mc{H}_E)$ is given by $\Phi^c (\rho) = \Tr_B (V \rho V^{\dagger}).$ The diamond norm of a linear map $\Phi: B(\mc{H}_A) \to B(\mc{H}_B)$ is defined as $\norm{\Phi}_{\diamond} \coloneqq \sup \limits_{\norm{M}_1 \leq 1} \norm{\Phi (M_{RA})}_1$, with supremum over all $M \in B(\mc{H}_{RA})$ with $d_R = d_A$ and $\norm{M}_1 \leq 1.$ \\

In this work, we will interchangeably refer to $n$, the number of channels in composition in $(\Xi^n)_{n \in \mbb{N}}$ as a time step parameter. They are equivalent as $(\Xi^n)_{n \in \mbb{N}}$ can also be viewed a discrete Quantum Markov Semigroup (dQMS) modelling the time evolution of an open quantum system as considered in \cite{singh2024information}. A dQMS is defined as follows.

\begin{definition} \label{def: dQMS}
    A discrete Quantum Markov Semigroup (dQMS) on a finite dimensional Hilbert space $\mc{H}$ is a one-parameter family $(\Phi_t)_{t \in \mbb{N}}$ of linear operators $\Phi_t : B(\mc{H}) \to B(\mc{H})$ satisfying: (i) $\Phi_{t + s} = \Phi_t \circ \Phi_s$ for every $t, s \in \mbb{N}$, (ii) $\Phi_0 = \text{id}_{\mc{H}}$, and (iii) $\Phi_t$ is completely positive and trace-preserving for every $t \in \mbb{N}.$
\end{definition}

\noindent A dQMS associated with a quantum channel $\Xi: B(\mc{H'}) \to B(\mc{H'})$ is the sequence $(\Xi^n)_{n \in \mbb{N}}.$ The justification for the nomenclature follows from $\Xi^n$ being \textit{discrete} as $n \in \mbb{N}$, \textit{Markovian} as interactions with an environment are assumed to be weak, and a \textit{semigroup} as it satisfies the semigroup property that $\forall ~r, s \in \mbb{N}: \Xi^{r + s} = \Xi^r \circ \Xi^s.$

\subsection{Statistical measures} \label{sec: stat}

In this work, we shall use the following statistical quantities defined here. The Kullback-Leibler divergence, $D (\cdot \big\Vert \cdot)$ is a type of statistical distance defined as $D (P \big\Vert Q) \coloneqq \sum \limits_{x \in \mc{X}} P(x) \log\left(\frac{P(x)}{Q(x)}\right)$ for probability distributions $P(x)$ and $Q(x).$ It is a measure of information lost when $Q(x)$ is used to approximate $P(x)$ \cite{Cover:2005lom}.

Often, in probability theory, one needs to bound how far the outcome of a random variable $X$ can deviate from its expected value. Bounds on such large deviations are known as \textit{tail bounds}. We shall be concerned with a certain tail bound for the binomial distribution known as the \textit{Chernoff bound} \cite{Arratia1989TutorialOL}, stated as follows.

Let $F (k; n, p)$ denote the cumulative distribution function for a random variable $X$ following the binomial distribution with parameters $n \in \mbb{N}$, $p \in \left[0, 1\right]$. The probability of getting exactly $k$ successes in $n$ independent Bernoulli trials with rate $p$ is given by the probability mass function $f(k, n, p) = \text{Pr}(X = k) = {n \choose k} p^k (1 - p)^{n - k},$ for $k = 0, 1, \ldots, n$ where ${n \choose k} = \frac{n!}{k! (n - k)!}.$ The cumulative distribution function is then expressed as $F (k; n, p) = \text{Pr}(X \leq k) = \sum \limits_{i = 0}^{\lfloor k \rfloor} {n \choose i} p^i (1 - p)^{n - i}$ where $\lfloor k \rfloor$ is the floor function quantifying the greatest integer less than or equal to $k$. For $k \leq np$ the Chernoff bound gives an upper bound on the lower tail of the cumulative distribution function, i.e., the probability that there are at most $k$ successes. The Chernoff bound states $$F(k; n, p) = \text{Pr}(X \leq k) \leq e^{- n D\left(\frac{k}{n} \Vert p\right)}$$ where $D\left(\cdot \Vert \cdot\right)$ is the Kullback-Leibler divergence.

\subsection{Quantum capacities}

We recall that the von Neumann entropy of a quantum state $\rho$ is defined as $S(\rho) \coloneqq - \Tr(\rho \log \rho).$ The coherent information \cite{Bar98} is an entropic quantity denoted by $Q^{(1)}$ characterising the amount of quantum information that can be transmitted over a quantum channel $\Phi$. It is defined as $$Q^{(1)} (\Phi) \coloneqq \max \limits_{\rho} I_c (\Phi, \rho)$$ where $I_c \left(\Phi, \rho\right) = S\left[(\Phi(\rho)\right] - S\left[\Phi^c(\rho)\right]$ is the entropy exchange between the channel and its environment. From the seminal work of \cite{Lloyd, Shor, Devetak}, we know that the quantum capacity of a channel is given by $Q(\Phi) \coloneqq \lim \limits_{n \to \infty} \frac{Q^{(1)}\left(\Phi^{\otimes n}\right)}{n}$. Moreover, the coherent information is superadditive \cite{smith2008quantum}, hence, it lower bounds the quantum capacity. In this work, we shall be concerned with the following definitions, especially for \cref{thm: continuity-capacity}.

\begin{definition} \label{def: segment-channel}
    Let $\Xi: M_2 \to M_2, ~\mc{E}: M_2 \to B(\mc{H}_A), ~\mc{D}: B(\mc{H}_B) \to M_2$, and $\mc{N}: B(\mc{H}_A) \to B(\mc{H}_B)$ be quantum channels such that $\Xi = \mc{D} \circ \mc{N} \circ \mc{E}$. Here, $\mc{E}$ and $\mc{D}$ represent encoding and decoding operations, while $\mc{N}$ represents any noisy channel. $\mc{H}_A, ~\mc{H}_B$ are most generally separable. We consider transmitting information over an n-fold composition $\left(\Xi^n\right)_{n \in \mbb{N}}.$
\end{definition}

\begin{remark} \label{rem: qubit-id}
    Let $\text{id}_2: M_2 \to M_2$ denote the identity channel on $M_2$. Its action on a state $\rho \in D(\mc{H})$ in Kraus representation is $\text{id}_2(\rho) = \rho$. Since this channel sends the entire input to the receiver, the environment does not receive any information. It is easy to see that the action of the complementary channel $\text{id}^c_2 (\rho) = \Tr(\rho)$, and hence $I_c(\text{id}_2, \rho) = S(\rho)$ which implies that the coherent information $Q^{(1)}(\text{id}_2) = \max \limits_{\rho} S(\rho) = \log 2 = 1.$
\end{remark}

\paragraph{Continuity of capacities} The continuity of capacities of quantum channels was established by Leung and Smith \cite{leung2009continuity}. We use a refinement that is tight due to Shirokov \cite{shirokov2017tight}.

\begin{lemma} \cite[Proposition 30]{shirokov2017tight} \label{lem: continuity}
    Let $\Phi$ and $\Psi$ be channels from $A$ to $B$. Then
    \begin{equation}
        \abs{Q(\Phi) - Q(\Psi)} \leq ~2\ve \log d_B + g(\ve)
    \end{equation}
    where $\ve = \frac{1}{2} \norm{\Phi - \Psi}_{\diamond}$, $d_B = \text{dim}(\mc{H})$, and $g(\ve) = (1 + \ve) h \left(\frac{\ve}{1 + \ve}\right).$ Here $h(\cdot)$ denotes the binary entropy given by $h(\ve) \coloneqq - \ve \log(\ve) - (1 - \ve) \log(1 - \ve).$
\end{lemma}

\begin{remark} \label{rem: shir-Q1}
    We note that the proof of \cite[Proposition 30]{shirokov2017tight} follows by bounding arbitrary n-shot capacities, not necessarily asymptotic; hence in the notation of \cref{lem: continuity}, the following holds $$ \abs{Q^{(1)}(\Phi) - Q^{(1)}(\Psi)} \leq ~2\ve \log d_B + g(\ve). $$
\end{remark}

\subsection{Spectral properties of channels}

We outline certain properties of channels originating from the structure of its spectrum that are used in this work. They are important for following \cref{thm: rad-conv}. 

\begin{definition} \label{def: T-matrix}
    Following \cite{ruskai2002analysis}, we define the $\mathbf{T}$-matrix of a channel $\Phi: M_2 \to M_2$ as a $4 \times 4$ matrix $T_{\Phi}$ with elements $t_{mn} = \ang{P , \Phi (Q)}: P, Q \in \{I, X, Y, Z\},$ where $X, Y, Z$ are the Pauli matrices, $\ang{\cdot , \cdot}$ denotes the normalised Hilbert-Schmidt inner product, and $I$ is the identity matrix in two dimensions.    
\end{definition}

\noindent Under the inner product in \cref{def: T-matrix}, the Pauli matrices form an orthonormal basis.

\begin{remark}
    The spectral radius of a channel $\Phi$ is the largest magnitude of its eigenvalues. Gelfand's formula \cite{Bla06} states that the spectral radius of an operator $A$ is given by $\lim \limits_{n \to \infty} \norm{A^n}^{1/n}.$
\end{remark}

\begin{lemma}\cite{ruskai2002analysis} \label{lem: ruskai-Tmat-qubit-ch}
    Let $\Xi: M_2 \to M_2$ be a qubit channel. Then the
    $\mathbf{T}$-matrix of $\Xi$, up to unitary conjugation, is always of the form
    \begin{equation} \label{eq: Tseg}
        T_{\Xi} = \begin{pmatrix}
            1 & 0 & 0 & 0 \\
            t_1 & \lambda_1 & 0 & 0 \\
            t_2 & 0 & \lambda_2 & 0 \\
            t_3 & 0 & 0 & \lambda_3
        \end{pmatrix}.
    \end{equation} \\    
\end{lemma}

\noindent The statement in \cref{def: T-matrix} is the natural representation of a quantum channel, see \cite[Section 2.2.2]{watrous2018theory}.

\begin{definition} \label{def: comp-mixed-ch}
    Let $\Xi_{\infty} \coloneqq \sum \limits_{\lvert \lambda_i \rvert = 1} \lambda_i \mc{P}_i$ be the CPTP map such that $\mc{P}_i$ are the spectral projectors corresponding to eigenvalues of $\Xi$ with unit magnitude, where $\Xi$ is a quantum channel as in \cref{def: segment-channel}. $\Xi_{\infty}$ arises as a limit point of $(\Xi^n)_{n \in \mbb{N}}$ \cite{szehr2015spectral}.
\end{definition}

\begin{lemma} \cite[Lemma III.1]{szehr2015spectral} \label{lem: wolf-spec-rad}
     Let $\mu$ be the spectral radius of $\Xi - \Xi_{\infty}$ and $R_n$ be the radius of convergence of $\left\Vert T_{\Xi}^n - T_{\Xi_{\infty}}^n \right\Vert$. Then,
     \begin{enumerate}
        \item \begin{equation} \label{eq: radius-of-conv}
                \left\Vert T_{\Xi}^n - T_{\Xi_{\infty}}^n \right\Vert \leq \left(\frac{1 + \mu}{2}\right)^n \coloneqq R_n, ~\text{for sufficiently large} ~n.
              \end{equation}
        \item \begin{equation} \label{eq: op-norm-T}
                \left\Vert T_{\Xi}^n - T_{\Xi_{\infty}}^n \right\Vert = \left\Vert \left(T_{\Xi} - T_{\Xi_{\infty}}\right)^n \right\Vert, ~\forall n.
              \end{equation} 
    \end{enumerate}
\end{lemma}

\begin{remark} \label{rem: n-spec-rad}
    For finite $n$, we have from \cite{szehr2015spectral} the following: $$ \left\Vert \Xi^n - \Xi_{\infty}^n \right\Vert = \left\Vert T_{\Xi}^n - T_{\Xi_{\infty}}^n \right\Vert \leq k \mu^n,$$ where $k$ depends on the spectrum of $\Xi$, on $n$, and on the dimension of the input Hilbert space of $\Xi$.
\end{remark}

\subsection{Quantum error correction} \label{sec: qecc}

In this work, we are interested in the transmission of information over a sequence of channels $\Xi^n$ as in \cref{def: segment-channel}. Notably, we allow the repetitive action of noise followed by error correction in succession. The following definitions are useful to follow \cref{thm: error-bound}.

\begin{definition} \label{def: encoder} 
    Let $\mathcal{C} = \spn\{\ket{0}_L, \ket{1}_L\}$ be a codespace encoding a logical qubit using the encoding channel $\mc{E}$ as in \cref{def: segment-channel} with the action $\mc{E} (\sigma) = S \sigma S^{\dagger},$ where $S = \ket{0}_L\bra{0}_A + \ket{1}_L\bra{1}_A$. $~~\ket{\cdot}_L$ denotes the logical basis states of the codespace, and $\ket{\cdot}_A$ denotes a physical state in the encoding quantum system. 
\end{definition}

\begin{definition} \label{def: recov}
    In quantum error correction \cite{knill1996theory}, one typically wants to know the action of a noisy channel $\mc{N}$ on an input state $\rho$, and find a recovery operation. A valid recovery operation denoted by $\mc{R}$ is a quantum channel $\mc{R}$ such that $\frac{1}{2} \norm{\mc{R} \circ \mc{N} (\rho) - \rho}_1 = 0.$ This is the case of perfect recovery. A generalised notion of \textit{approximate recovery} also exists \cite{leung1997approximate} wherein $\frac{1}{2} \norm{\mc{R} \circ \mc{N} (\rho) - \rho}_1 \approx O(\lambda^2)$, for some noise parameter $\lambda.$ 
\end{definition}

\begin{remark} \label{eq: recovery}
    Let $\mc{R}$ be a quantum channel that is a valid recovery operation for any noisy channel $\Phi$ acting on an encoded state $\mc{E}(\sigma)$. Since, $\mc{E}$ is an isometry, we can write
    \begin{align*}
        \frac{1}{2}\left\Vert \mc{D} \circ \Phi \circ \mc{E} (\sigma) - \sigma \right\Vert_1 &= \frac{1}{2}\left\Vert \mc{E} \circ \mc{D} \circ \Phi \circ \mc{E}(\sigma) - \mc{E}(\sigma) \right\Vert_1 \\
        &= \frac{1}{2}\left\Vert \mc{R} \circ \Phi (\rho) - \rho \right\Vert_1
    \end{align*} 
    where in the last equality we use $\mc{R} = \mc{E} \circ \mc{D}$ and $\mc{E}(\sigma) = \rho$ for some input state $\sigma \in \mc{H}_A$ from the physical encoding system. Another way to write an encoded state $\rho$ is to consider $\rho = \ket{\psi}\bra{\psi}$, where $\ket{\psi} = \alpha \ket{0}_L + \beta \ket{1}_L.$
\end{remark}

\section{Main Results}  \label{sec: main}

If we wish to send quantum data over a sequence of channels $\left(\Xi^n\right)_{n \in N}$ by performing an error correction procedure with an encoding-decoding pair $\left(\mc{E}, \mc{D}\right)$ at each time step, $n$ to protect from the action of noise $\mc{N}$ , then we first need to know how effective the code is. This is quantified by $\ve$ as the closeness of $\Xi = \mc{D} \circ \mc{N} \circ \mc{E}$ to the identity channel in diamond norm. Next, we need to know the evolution of $\Xi^n$ in time, i.e., the distance of $\Xi^n$ from the identity map as $\Xi$ evolves with $n$. We see that we cannot do better than $n\ve$, in general. Given that this holds, and the fact that quantum capacities are continuous, we can exactly specify how much coherent information is preserved at each node $n$, for $\Xi^n$. This idea is made precise in the following theorem. Moreover, since the coherent information is an achievable rate for distillable entanglement \cite{devetak2005distillation}, any non-trivial lower bound indicates that entanglement can be generated between the $n=1$ and $n=n_{\text{final}}$ nodes by transmission over the sequence $\Xi^n$, provided $Q^{(1)}(\Xi^{n_{\text{final}}}) > 0.$

\begin{theorem} \label{thm: continuity-capacity}
    Let $\Xi$ be a quantum channel as in \cref{def: segment-channel}. Let $\ve \in \left[0, 1\right]$. For any nonnegative integer $n$, if 
    \begin{equation} \label{eq: dist-id}
        \frac{1}{2}\left\Vert \Xi - \text{id}_2 \right\Vert_{\diamond} \leq \ve,
    \end{equation}
    then
    \begin{equation} \label{eq: cont-cap}
        Q^{(1)} (\Xi^n) \geq 1 - 2n\ve - (1 + n\ve) h \left(\frac{n\ve}{1 + n\ve}\right).
    \end{equation}    
\end{theorem}

\begin{proof}
    First we note that $\text{id}_2^n (\rho) = \text{id}_2 (\rho).$ Given $\left\Vert \Xi - \text{id}_2 \right\Vert_{\diamond} \leq \ve$, then
    \begin{align*}
        \left\Vert \Xi^n - \text{id}_2 \right\Vert_{\diamond} &= \left\Vert \Xi^n - \Xi^{n - 1} + \Xi^{n - 1} - \Xi^{n - 2} + \Xi^{n - 2} - \ldots + \Xi - \text{id}_2 \right\Vert_{\diamond} \\
        &= \left\Vert \Xi^{n - 1} \left(\Xi - \text{id}_2\right) + \Xi^{n - 2}\left(\Xi - \text{id}_2\right) + \ldots + \Xi - \text{id}_2 \right\Vert_{\diamond} \\
        &\leq \left\Vert \Xi^{n - 1} \left(\Xi - \text{id}_2\right)\right\Vert_{\diamond} + \left\Vert \Xi^{n - 2}\left(\Xi - \text{id}_2\right)\right\Vert_{\diamond} + \ldots + \left\Vert\Xi - \text{id}_2 \right\Vert_{\diamond} \\
        &\leq \left\Vert \Xi - \text{id}_2 \right\Vert_{\diamond} + \left\Vert \Xi - \text{id}_2 \right\Vert_{\diamond} + \ldots + \left\Vert\Xi - \text{id}_2 \right\Vert_{\diamond} \\
        &\leq 2 n\ve
    \end{align*}
    where the first and second equality is a telescoping argument, the first inequality is by triangle inequality, the second inequality is by data processing inequality, and the final inequality is by the given condition \eqref{eq: dist-id}. Now, applying \cref{rem: shir-Q1,lem: continuity} to $\frac{1}{2} \left\Vert \Xi^n - \text{id}_2 \right\Vert_{\diamond} \leq n\ve$ and noting that $Q^{(1)}(\text{id}_2) = 1$ (see \cref{rem: qubit-id}) the result follows.
\end{proof}

\begin{remark} \label{rem: Shir-gen}
    In the above theorem, we showed the case for $\Xi$ being a qubit channel as in \cref{def: segment-channel}, as this was the case of interest with non trivial lower bounds possibly existing for higher dimensional channels in \cite{singh2024information}, although not as a function of $n$. We note that our proof identically holds for any $\Xi$ with finite output dimension. This is because Shirokov's continuity bound (\cref{lem: continuity}) holds for any $\Xi': B({\mc{\bar{H}}_A)} \to B(\mc{\bar{H}}_B)$, where $\text{dim}(\mc{\bar{H}}_B), ~\text{dim}(\mc{\bar{H}}_{B'}) < \infty$, but $\text{dim}(\mc{\bar{H}}_A), ~\text{dim}(\mc{H}_{A'}) = \infty$ such that $\Xi' = \mc{D'} \circ \mc{N'} \circ \mc{E'}$, for some quantum channels $\mc{E'} : B(\mc{\bar{H}}_A) \to B(\mc{\bar{H}}_{A'}), ~\mc{D'} : B(\mc{\bar{H}}_{B'}) \to B(\mc{\bar{H}}_B), ~\mc{N'} : B(\mc{\bar{H}}_{A'}) \to B(\mc{\bar{H}}_{B'})$. Since $Q^{(1)}(\text{id}_B) = \log d_B$, where $d_B = \text{dim}(\mc{\bar{H}}_B)$, follows trivially from the argument in \cref{rem: qubit-id}, this leads to the following corollary.
\end{remark}

\begin{corollary} \label{cor: gen-dim-cont-cap}
    Under the conditions of \cref{rem: Shir-gen}, if $\frac{1}{2} \norm{\Xi' - \text{id}_B}_{\diamond} \leq \ve$, then
    $Q^{(1)} (\Xi'^n) \geq \log d_B(1 - 2n\ve) - (1 + n\ve) h \left(\frac{n\ve}{1 + n\ve}\right).$
\end{corollary}

Now that we know how much entanglement we can generate over the network $(\Xi^n)_{n \in \mbb{N}}$ as a function of $n$ by performing error correction at each time step, we would like to know if there are any advantages coming from the structure of a single node $\Xi = \mc{D} \circ \mc{N} \circ \mc{E}$. It turns out that how well $\Xi^n$ can preserve information over time, or equivalently how much distance can a network with identical nodes generate entanglement over is simply a function of the spectrum of a single node. This is what the first claim in the following theorem quantifies. The second claim states the condition on the eigenvalues of $\Xi$ under which the network $\Xi^n$ rapidly loses any quantum data sent. The third claim states a condition on the eigenvalues of $\Xi$ and the time step $n$, such that the network $\Xi^n$ is able to preserve information until that time step. All these claims follow from studying the rate at which $\Xi^n$ converges to its limit point $\Xi_{\infty}$ (see \cref{def: comp-mixed-ch}) which describes its asymptotic image (\cite{szehr2015spectral}), a noisy channel that is completely mixing, in the sense that it bears no information about the input as it maps every input to the same output. In this result, we restrict to sending only a logical qubit over $\Xi^n$ due to the convenience in expressing the $\mathbf{T}-$matrix of qubit channels (see \cref{def: T-matrix} and \cref{lem: ruskai-Tmat-qubit-ch}) which is unknown for higher dimensions (\cite{ruskai2002analysis}). 

\begin{theorem} \label{thm: rad-conv}
    Let $\Xi: M_2 \to M_2$ be a qubit channel with $T_{\Xi}$ as in \cref{eq: Tseg}. Let $\lambda_i, t_i \in \left[0, 1\right) ~\forall i.$  Let $\mu$ be the spectral radius of $\Xi - \Xi_{\infty}$ and $R_n$ be the radius of convergence of $\Vert T_{\Xi}^n - T_{\Xi_{\infty}}^n \Vert$ as in \cref{eq: radius-of-conv}. Let $\ve, \delta \in \left(0, 1\right]$. Then,
    \begin{enumerate}
        \item \begin{equation}
                \mu = \max \{0, \lambda_1, \lambda_2, \lambda_3\}.
              \end{equation}
        \item \begin{equation}
                \lim \limits_{n \to \infty} R_n = 0, ~~~~~\text{if} ~\forall i ~~\lambda_i < 1.
              \end{equation} 
        \item For $\ve, \delta \leq 1$, if $\mu \geq 1 - \ve,$ then $R_n \geq 1 - \delta$ as long as $n \leq \frac{2 \delta}{\ve}.$ \\
    \end{enumerate}    
\end{theorem}

\begin{remark}
    Note that channels in \cref{def: segment-channel} are a special case of the above.
\end{remark}

\begin{proof}[Proof of \cref{thm: rad-conv}] 
    We can always write $T_{\Xi} = S J_{\Xi} S^{-1}$, where $J_{\Xi}$ is the Jordan canonical form of $T_{\Xi}$ and $S$ is the similarity transformation matrix. Computing $J_{\Xi}$ and $S$ we get,
    \begin{equation} \label{eq: JC-Xi}
        J_{\Xi} = \begin{pmatrix}
            1 & 0 & 0 & 0 \\
            0 & \lambda_1 & 0 & 0 \\
            0 & 0 & \lambda_2 & 0 \\
            0 & 0 & 0 & \lambda_3
        \end{pmatrix},   
    \end{equation}

    \begin{equation} \label{eq: S-Xi}
        S = \begin{pmatrix}
            \frac{1 - \lambda_3}{t_3} & 0 & 0 & 0 \\
            \frac{t_1 (\lambda_3 - 1)}{t_3 (\lambda_1 - 1)} & 1 & 0 & 0 \\
            \frac{t_2 (\lambda_3 - 1)}{t_3 (\lambda_2 - 1)} & 0 & 1 & 0 \\
            1 & 0 & 0 & 1
        \end{pmatrix}.    
    \end{equation}
    
    By \cref{def: comp-mixed-ch} and since $\lambda_i < 1, ~~\forall i$, the Jordan canonical form of $T_{\Xi_{\infty}}$ is
    \begin{equation} \label{eq: JC-Xi-infty}
        J_{\Xi_{\infty}} = \begin{pmatrix}
            1 & 0 & 0 & 0 \\
            0 & 0 & 0 & 0 \\
            0 & 0 & 0 & 0 \\
            0 & 0 & 0 & 0
        \end{pmatrix}.    
    \end{equation}
    
    Then, $T_{\Xi_{\infty}} = S J_{\Xi_{\infty}} S^{-1}$ is given by

    \begin{equation} \label{eq: T-Xi-infty}
        T_{\Xi_{\infty}} = \begin{pmatrix}
            1 & 0 & 0 & 0 \\
            \frac{t_1}{1 - \lambda_1} & 0 & 0 & 0 \\
            \frac{t_2}{1 - \lambda_2} & 0 & 0 & 0 \\
            \frac{t_3}{1 - \lambda_3} & 0 & 0 & 1
        \end{pmatrix}.    
    \end{equation}   

    Now we calculate $T_{\Xi}^n$ and $T_{\Xi_{\infty}}^n$ and define $\Delta_n \coloneqq T_{\Xi}^n - T_{\Xi_{\infty}}^n$. By \cref{eq: op-norm-T} in \cref{lem: wolf-spec-rad}, $\left\Vert \Delta_n \right\Vert = \left\Vert \left(T_{\Xi} - T_{\Xi_{\infty}}\right)^n \right\Vert$, hence calculating eigenvalues of $\Delta_n$ suffices to the find $\mu$, the spectral radius of $\Xi - \Xi_{\infty}$. Computing, we get that the eigenvalues of $\Delta_n$ are $\{0, \lambda_1^n, \lambda_2^n, \lambda_3^n\}$. Applying Gelfand's formula $\lim \limits_{n \to \infty} \left\Vert \Delta_n \right\Vert^{1/n} = \mu$, we get that $\mu = \max \{0, \lambda_1, \lambda_2, \lambda_3 \}.$

    Next, by \cref{eq: radius-of-conv} in \cref{lem: wolf-spec-rad}, and since $\lambda_i \geq 0$, it easily follows from the fact if $\forall i,~ \lambda_i < 1$, then $\mu \leq 1$, hence $ 0 \leq \lim \limits_{n \to \infty} R_n \leq \lim \limits_{n \to \infty}\left(\frac{1 + \mu}{2}\right)^n = 0$.
    
    Finally, let $\mu \geq 1 - \ve$, for $\ve \leq 1$. Then $R_n = \left(\frac{1 + 1 - \ve}{2}\right)^n = \left(1 - \frac{\ve}{2}\right)^n \geq 1 - \frac{n \ve}{2}.$ Let $\delta \leq 1$ such that $\delta \geq \frac{n \ve}{2}.$ Then $R_n \geq 1 - \delta$ whenever $n \leq \frac{2 \delta}{\ve}.$ This completes the proof.
    
\end{proof}

In the theory of quantum error correction \cite{knill1996theory}, the correctability of a code is restricted to the error set it can correct. Hence, different choices of a code $(\mc{E}, \mc{D})$ against noise $\mc{N}$ correct different errors, and possibly do not correct others not contained in the error set. Various constructions of quantum networks using error correction so designed thus have varied capabilities. However, it is a priori not clear if one might be able to sufficiently correct most of the error occurring during transmission over the network $(\Xi^n)_{n \in \mbb{N}}$ so as to preserve non trivial $Q^{(1)}(\Xi^n),$ until certain $n > 0.$ This is because it is not a property of the error set that a particular code can correct. To attempt to understand this, one must be able to keep track of the total error in transmission, and hence bound it as a function of the noise channel; thus independently of the chosen code. The following theorem achieves this. We bound the total error occurring at a single node $\Xi = \mc{D} \circ \mc{N} \circ \mc{E}$ in the sequence $\Xi^n$ as a function of the arbitrary noise model $\mc{N}$, relying only on the existence of a recovery operation correcting a certain number of errors (see \cref{def: recov}). Note that this is more general than choosing a specific code and assessing its performance against noise, as it quantifies the amount of error that can accumulate at each node $\Xi$ owing only to the noise model.

\begin{theorem} \label{thm: error-bound}
    Let $\Phi (\rho) = \sum \limits_{i = 1}^d M_i \rho M_i^{\dagger}$ be a noisy quantum channel with a recovery operation $\mc{R}$ that corrects errors $\{M_1, \ldots, M_k\}$ for some $k \leq d$ acting on an encoded state $\rho$. Then,
    \begin{equation} \label{eq: err-bd}
        \frac{1}{2}\left\Vert \left(\mc{R} \circ \Phi \right)(\rho) - \rho \right\Vert_1 \leq \left\Vert \sum \limits_{i = k+1}^d M_i^{\dagger}M_i \right\Vert.
    \end{equation}
\end{theorem}

\begin{proof}
    For $i \leq k, ~\mc{R} \left( M_i \rho M^{\dagger}_i\right) = \Tr\left(M_i \rho M_i^{\dagger}\right)\rho$ since $\mc{R}$ is a valid recovery operation for $\{M_1, \ldots M_k\} : k \leq d$. Then
    \begin{align*}
        \left(\mc{R} \circ \Phi\right) (\rho) &= \Tr \left(\sum \limits_{i = 1}^k M_i \rho M_i^{\dagger} \right) \rho + \mc{R} \left(\sum \limits_{i = k + 1}^d M_i \rho M_i^{\dagger} \right) \\
        &= \rho - \Tr \left(\sum \limits_{i = k + 1}^d M_i \rho M_i^{\dagger} \right) \rho + \mc{R} \left(\sum \limits_{i = k + 1}^d M_i \rho M_i^{\dagger} \right).
    \end{align*}
    
    Therefore,
    \begin{align*}
        \frac{1}{2} \left\Vert \mc{R} \circ \Phi \left(\rho\right) - \rho \right\Vert_1 &\leq \frac{1}{2}\left\Vert \Tr \left(\sum \limits_{i = k + 1}^d M_i \rho M_i^{\dagger}\right) \rho \right\Vert_1 + \left\Vert \mc{R} \left( \sum \limits_{i = k + 1}^d M_i \rho M_i^{\dagger} \right) \right\Vert_1 \\
        &= \frac{1}{2} \Tr \left( \sum \limits_{i = k + 1}^d M_i \rho M_i^{\dagger}\right) + \frac{1}{2} \Tr \left( \mc{R} \left( \sum \limits_{i = k+ 1}^d M_i \rho M_i^{\dagger} \right) \right) \\
        &= \Tr \left( \sum \limits_{i = k + 1}^d M_i \rho M_i^{\dagger} \right) \\
        &= \Tr \left( \sum \limits_{i = k + 1}^d M_i^{\dagger} M_i \rho \right) \\
        &\leq \left\Vert \sum \limits_{i = k + 1}^d M_i^{\dagger}M_i \right\Vert
    \end{align*}
    where the first line is triangle inequality, the second line is the $L_1$ norm for positive operators, the third line is the trace preserving property of channels, the fourth line is cyclicity of the trace, and the final line is by the definition of the operator norm. This gives the result.
\end{proof}

\section{Applications} \label{sec: app}

\subsection{Infinite dimensional pure-loss channel}

In this section, we consider the noise model of arbitrary excitation loss described by the infinite dimensional pure-loss channel (see \cref{def: loss-channel}) which is believed to be the dominant source of noise in quantum networks. We quantify the error bound of \cref{thm: error-bound} for this channel in \cref{cor: loss-err-bd}.

\begin{definition} \label{def: loss-channel}
    The bosonic pure-loss channel $\mc{M}_{\eta}: B(L^2(\mbb{R})) \to B(L^2(\mbb{R}))$ is an infinite dimensional quantum channel. Its action in Kraus form is $\mc{M}_{\eta}(\rho) = \sum \limits_{l = 0}^{\infty} A_l \rho A_l^{\dagger}$, where $A_l = \sqrt{\frac{(1 - \eta)^l}{l!}} \sqrt{\eta}^{a^{\dagger}a} a^l$. $a^{\dagger}, a$ are the bosonic creation and annihilation operators satisfying the commutation relation $\left[a, a^{\dagger}\right] = \mathds{1}$. $\eta \in \left[0, 1\right]$ is the transmission probability. It models the physical phenomenon of arbitrary excitation loss, and accurately describes attenuation in a linear quantum network modelled as $(\Xi^n)_{n \in \mbb{N}}$ in \cref{def: segment-channel}. \\
\end{definition}

\begin{corollary} \label{cor: loss-err-bd}
    Let $\ket{d}$ be the Fock state with the highest occupation number in $\mc{H}_A$ such that the action of the bosonic pure-loss channel on an encoded state is $\mc{M}_{\eta}^d (\rho) = \sum \limits_{l = 1}^{d} A_l \rho A_l^{\dagger}$ with there existing a recovery operation $\mc{R}$ correcting errors $\{A_1, \ldots, A_k\}$ for some $k \leq d$, where $A_l$ is as in \cref{def: loss-channel}. Then,
    \begin{equation}
        \frac{1}{2}\left\Vert \left(\mc{R} \circ \mc{M}_{\eta}^d \right)(\rho) - \rho \right\Vert_1 \leq \max \limits_{m = k + 1, \ldots, d} e^{-m D\left(\frac{k+1}{m} \big\Vert \eta\right)}
    \end{equation}  
    where $D (\cdot \big\Vert \cdot)$ is the Kullback-Liebler divergence (see \cref{sec: stat}).
\end{corollary}

\begin{proof}
    By \cref{thm: error-bound} $$\frac{1}{2}\left\Vert \left(\mc{R} \circ \mc{M}_{\eta}^d \right)(\rho) - \rho \right\Vert_1 \leq \left\Vert \sum_{l = k + 1}^d A_l^{\dagger} A_l \right\Vert.$$
    Now, using \cref{def: loss-channel}, we can write $E \coloneqq \sum \limits_{l = k + 1}^d A_l^{\dagger} A_l = \sum \limits_{l = k + 1}^d \frac{\left(1 - \eta\right)^l}{l!} {a^{\dagger}}^{l} \eta^{a^{\dagger} a} a^{l}$. 
    Then, 
    \begin{align*}
        E \ket{m} &= \sum \limits_{l = k + 1}^d \frac{\left(1 - \eta\right)^l}{l!} \frac{m!}{\left(m - l\right)!} \eta^{m - l} \ket{m} \\
        &= \sum \limits_{l = k + 1}^m \binom{m}{l} \eta^{m - l} \left( 1 - \eta\right)^l \ket{m}.
    \end{align*}
    Therefore,
    \begin{align*}
        \left\Vert E \right\Vert &= \max \limits_{m = k+1, \ldots, d} \sum \limits_{l = k + 1}^m \binom{m}{l} \eta^{m - l} \left( 1 - \eta\right)^l \\
        &\leq \max \limits_{m = k + 1, \ldots, d} e^{-m D\left(\frac{k+1}{m} \big\Vert \eta\right)}
    \end{align*}
    where the last inequality is the Chernoff tail bound for the binomial distribution (see \cref{sec: stat}). This gives the result.
\end{proof}

\subsection{Amplitude damping channel and its generalisations}

In this section, we apply \cref{thm: continuity-capacity,thm: rad-conv,thm: error-bound} to the amplitude damping channel and its bosonic counterpart to exemplify the non triviality of the capacity bounds obtained on the sequence of channels $(\Xi^n)_{n \in \mbb{N}}$ as a function of the time step parameter $n$.

The qubit amplitude damping channel denoted by $\mc{A}_{\gamma}: M_2 \to M_2$ where $\gamma \in \left[0, 1\right]$ is the damping parameter or the probability of losing one excitation. This is a non unital and extremal channel \cite{ruskai2002analysis}, and physically models the loss of one excitation to the environment from the system. The Kraus representation of $\mc{A}_{\gamma} (\rho) = A_0 \rho A_0^{\dag} + A_1 \rho A_1^{\dag}$ is  
$A_0 = \begin{pmatrix}
       1 & 0 \\
       0 & \sqrt{1 - \gamma} \\
       \end{pmatrix}$
and $A_1 = \begin{pmatrix}
           0 & \sqrt{\gamma} \\
           0 & 0
           \end{pmatrix}.$
For $\gamma \in \left[0, \frac{1}{2}\right]$, the channel $\mc{A}_{\gamma}$ is \textit{degradable}, i.e., $\exists~$ a CPTP map $\mc{D}$ such that $\mc{D} \circ \mc{A}_{\gamma} = \mc{A}_{\gamma}^{c}$, where $\mc{A}_{\gamma}^{c}$ is complementary channel of $\mc{A}_{\gamma}.$ Hence the quantum capacity of $Q(\mc{A}_{\gamma})$ is equal to its coherent information \cite{cubitt2008structure}. For $\gamma \in \left[\frac{1}{2}, 1 \right]$, $\mc{A}_{\gamma}$ is \textit{anti-degradable}, i.e., $\exists~$ a CPTP map $\mc{L}$ such that $\mc{L} \circ \mc{A}^{c}_{\gamma} = \mc{A}_{\gamma}.$ In the regime where $\mc{A}_{\gamma}$ is anti-degradable, the quantum capacity $Q(\mc{A}_{\gamma}) = 0.$ From knowing the action of $\mc{A}_{\gamma}$ in its Kraus form, $T_{\mc{A}_{\gamma}}$ can be easily computed using \cref{def: T-matrix} as follows

\begin{equation} \label{eq: Tmat of amp-damp}
    T_{\mc{A}_{\gamma}} = \begin{pmatrix}
                          . 1 & 0 & 0 & 0 \\
                            0 & \sqrt{1 - \gamma} & 0 & 0 \\
                            0 & 0 & \sqrt{1 - \gamma} & 0 \\
                            \gamma & 0 & 0 & 1 - \gamma
                          \end{pmatrix}.
\end{equation}

Since we are interested in achieving successful entanglement distribution across a communication link $(\Xi^n)_{n \in \mbb{N}}$, we shall be concerned with the damping regime where $\mc{A}_{\gamma}$ is degradable so that the quantum capacity and hence the coherent information is positive. Therefore, we consider $\gamma \in \left[0, \frac{1}{2}\right].$ From \cref{thm: rad-conv}, the spectral radius of $\mc{A}_{\gamma} - \mc{A}_{{\gamma}_{\infty}}$ is $\mu = \sqrt{1 - \gamma},$ and therefore $\mu \in \left[\frac{1}{\sqrt{2}}, 1\right].$ Let there be some code $\left(\mc{E}, \mc{D}\right)$ such that $\frac{1}{2}\left\Vert \mc{D}\circ\mc{A}_{\gamma}\circ \mc{E} - \text{id}_2 \right\Vert_{\diamond} \leq \ve = 0.0005.$ For ease of notation, let $\mc{D}\circ\mc{A}_{\gamma}\circ \mc{E} = \Xi_{\mc{A}}$. Let us consider an arbitrarily chosen time step $n = 44$. Then by applying \cref{thm: continuity-capacity} we have that $Q^{(1)}(\Xi_{\mc{A}}^{44}) \geq 0.803.$ Again, from \cref{thm: rad-conv}, given $\ve = 0.0005$, $n = 44$ we find that $\delta \geq 0.011$ and hence $R_n \geq 0.989$ whenever $n \leq 44.$ Thus, we see a good agreement in the lower and upper bounds on the rate at which information is dissipated through $(\Xi_{\mc{A}}^n)_{n \in \mbb{N}}$ by applying the continuity bound on the one-shot capacity in \cref{thm: continuity-capacity} and the bound on the radius of convergence of the channel to its limit point in \cref{thm: rad-conv}. Moreover, we are able to estimate analytically the throughput of the communication link $(\Xi_{\mc{A}}^n)_{n \in \mbb{N}}$ undergoing amplitude damping noise at an arbitrarily chosen finite time step $n$. This signifies the efficacy of the bound on the one-shot capacity which is indeed non trivial, and together with the upper bound via the radius of convergence provides a reasonable estimate on how much information is transmitted as a function of $n.$

Next, consider the bosonic amplitude damping channel with an input Hilbert space that is infinite dimensional with a basis $\{ \ket{0}, \ket{1}, \ket{2}, \ket{3}, \ldots \}$ where $\ket{j}$ denotes a state with $j$ excitations. The bosonic generalisation of the amplitude damping channel is $\mc{B}_{\gamma}(\rho) = \sum \limits_{k=0}^{\infty} B_k \rho B_k^{\dag}$ where $B_k = \sum \limits_{j \geq k} \sqrt{\binom{j}{k}} \sqrt{{(1 - \gamma)}^{j - k} \gamma^k} \ket{j - k}\bra{j}$ represents the loss of $k$ excitations from the system. In particular, the first two Kraus operators are $B_0 = \sum \limits_{j} (1 - \gamma)^{\frac{j}{2}} \ket{j}\bra{j},$ and $B_1 = \sum \limits_{j \geq 1} \sqrt{j (1 - \gamma)^{j - 1} \gamma} \ket{j - 1}\bra{j}.$ The existence of a simple bosonic quantum error correcting code, the first of its kind, has long been known \cite{chuang1997bosonic} with codewords $\ket{0}_L = \frac{1}{\sqrt{2}}\left(\ket{40} + \ket{04}\right)$, $\ket{1}_L = \ket{22}$ exactly correcting the error set $\mc{F} = \{B_0 \otimes B_0, B_0 \otimes B_1, B_1 \otimes B_0\}.$ We want to quantify via \cref{thm: error-bound} the total error accumulating at the communication node by calculating the RHS of \cref{eq: err-bd} as $\left\Vert I - (B_0 \otimes B_0)^2 - B_0^2 \otimes B_1^{\dag}B_1 - B_1^{\dag}B_1 \otimes B_0^2 \right\Vert$ since $B_0$ is hermitian. We restrict the input Hilbert space to $j = 4$ excitations as that is the highest occupation number of a physical state in the codespace. 

For completeness, we explicitly show the calculation of $$\left\Vert I - (B_0 \otimes B_0)^2 - B_0^2 \otimes B_1^{\dag}B_1 - B_1^{\dag}B_1 \otimes B_0^2 \right\Vert.$$ Setting $j = 4$, we get $$B_0 = \ketbra{0} + \sqrt{1 - \gamma} \ketbra{1} + (1 - \gamma) \ketbra{2} + (1 - \gamma)^{3/2}\ketbra{3} + (1 - \gamma)^2 \ketbra{4}.$$ Since $B_0$ is hermitian, we can calculate $B_0^2$ and write it in the following compact form $$B_0^2 = \sum \limits_{m = 0}^4 (1 - \gamma)^m \ketbra{m}. $$ Then $$ B_0^2 \otimes B_0^2 = \sum \limits_{m, n = 0}^4 (1 - \gamma)^{m + n} \ketbra{mn}. $$ Next we compute $$B_1 = \sqrt{\gamma} \ketbra{0}{1} + \sqrt{2 (1 - \gamma) \gamma} \ketbra{1}{2} + \sqrt{3 \gamma (1 - \gamma)^2} \ketbra{2}{3} + \sqrt{4 (1 - \gamma)^3 \gamma} \ketbra{3}{4}. $$ Note that $B_1$ is not hermitian, so we need to calculate $B_1^{\dag}B_1$ explicitly. We do this and express in the following compact form $$B_1^{\dag}B_1 = \sum \limits_{m = 1}^4 m \gamma (1 - \gamma)^{m - 1} \ketbra{m} .$$ Now we can write down the remaining two terms we need as follows $$B_0^2 \otimes B_1^{\dag}B_1 = \sum \limits_{m, n = 0}^4 n \gamma (1 - \gamma)^{m + n - 1} \ketbra{m n}, $$ and $$ B_1^{\dag} B_1 \otimes B_0^2 = \sum \limits_{m, n = 0}^4 m \gamma (1 - \gamma)^{m + n - 1} \ketbra{m n}. $$ Using these expressions, the operator norm that we need to evaluate takes the form
\begin{align}
    &\left\Vert I - \sum \limits_{m, n =0}^4 (1 - \gamma)^{m + n} \left[1 + (m + n)\frac{\gamma}{1 - \gamma}\right] \ketbra{m n} \right\Vert \\
    &= \max \limits_{m, n = 0, \ldots, 4} 1 - (1 - \gamma)^{m + n} \left[ 1 + (m + n)\frac{\gamma}{1 - \gamma}\right] \\
    &= \max \limits_{p = 0, \ldots, 8} 1 - (1 - \gamma)^p \left[1 + \frac{p \gamma}{1 - \gamma}\right] \label{eq: maximise},
\end{align}

where we put $m + n = p.$ 

Now we can further simplify starting with
\begin{align}
    &1 - (1 - \gamma)^p \left[1 + \frac{p \gamma}{1 - \gamma} \right] \\
    &= 1 - (1 - \gamma)^{p - 1} \left[ 1 - \gamma + p \gamma \right] \\
    &= 1 - (1 - \gamma)^{p - 1} \left[1 + (p - 1) \gamma \right] \\
    &\leq 1 - (1 - (p - 1) \gamma)(1 + (p - 1) \gamma) \label{eq: binom} \\
    &= 1 - (1 - (p - 1)^2 \gamma^2 \\
    &= (p - 1)^2 \gamma^2 \label{eq: simplified},
\end{align}

where in \cref{eq: binom} we used that $(1 - \gamma)^{p - 1} \geq 1 - (p - 1) \gamma.$ 

Putting \cref{eq: simplified} into \cref{eq: maximise}, we get

\begin{align}
    &\left\Vert I - (B_0 \otimes B_0)^2 - B_0^2 \otimes B_1^{\dag}B_1 - B_1^{\dag}B_1 \otimes B_0^2 \right\Vert \\
    &\leq \max \limits_{p = 0, \ldots, 8} (p - 1)^2 \gamma^2 \\
    &= 49 \gamma^2 \label{eq: final-err-bd-amp-damp}, 
\end{align}

where we used that $(p - 1)^2 \gamma^2$ is a monotonically increasing function. 

We can now put the bound in \cref{eq: final-err-bd-amp-damp} into the continuity bound on the one-shot capacity in \cref{thm: continuity-capacity} with $\ve = 49 \gamma^2$ to get

\begin{equation} \label{eq: final-cpcty}
    Q^{(1)} (\Xi_{\mc{A}_{\gamma}}^n) \geq 1 - 98 n \gamma^2 - (1 + 49 n \gamma^2) h \left(\frac{49 n \gamma^2}{1 + 49 n \gamma^2}\right),
\end{equation}

where $\Xi_{\mc{A}_{\gamma}} = \mc{D} \circ \mc{A}_{\gamma} \circ \mc{E}$ and $(\mc{E}, \mc{D})$ corresponds to an encoder-decoder pair for the quantum error correcting code $\mc{C} = \spn\{\ket{0}_L, \ket{1}_L\}$ with codewords $\ket{0}_L = \frac{1}{\sqrt{2}}\left(\ket{40} + \ket{04}\right)$, $\ket{1}_L = \ket{22}$. To explicitly calculate $(\mc{E}, \mc{D})$ for this code, one may use the definitions in \cref{sec: qecc} or see \cite{chuang1997bosonic}. \\

Thus, we see for the amplitude damping noise model, how one can quantify, using our error and capacity bounds in \cref{thm: error-bound,thm: continuity-capacity}, the error accumulating at each node in the communication link $\Xi^n$ beyond what Knill-Laflamme theory \cite{knill1996theory} predicts; as a function of $n$, and the damping parameter $\gamma$ originating from the chosen noise model. We also compute the non trivial lower bound on the one-shot quantum capacity under finite error $49 \gamma^2$ in the amplitude damping noise model via \cref{eq: final-cpcty}.

\section{Conclusion} \label{sec: conc}

In summary, we study the problem of lower bounding the one-shot quantum capacity of an $n-$fold composition of channels $(\Xi^n)_{n \in \mbb{N}}$ in the sequential setting at short time scales. We quantify a lower bound that is non trivially a function of the time step $n$, thus showing that it is possible to preserve information over $\Xi^n$ up to certain finite time $n$ by performing quantum error correction at each time step. We quantify in simple terms how well this can be done by relating the rate of transmission to the spectrum of $\Xi$, and characterise the maximum possible error occurring as a function of the noise model. Our results suggest that the sequence of channels $\Xi^n$ is a useful model for a quantum network, analytically providing insights into its structure, and in turn determining its capabilities for distributing entanglement between spatially separated points. By applying our analysis to the infinite dimensional pure loss channel, as well as the amplitude damping channel and its bosonic counterpart, we characterise a relevant noise scenario for realistic networks.

\paragraph{Future directions} It would be interesting to consider inequivalent error correction mechanisms at each time step, and study the rate of transmission. In that case, the current model needs to be altered as it only captures properties of a network with identical nodes. Moreover, it is worth finding out if sharper bounds in the sequential setting can be obtained from a \textit{continuous-time} analysis (\cite{muller2014quantum}) as opposed to discrete Quantum Markov Semigroups considered so far. Further, one may want to extend our error bound at a single node to the entire sequence of channels which may require new analytical techniques. Among applications, other noise models for networks may be studied such as the bosonic dephasing channel (\cite{lami2023exact}) independently or in combination with pure-loss (\cite{mele2024quantum}), wherein the quantum capacity of the former was recently solved while that of the latter is still unknown.

\paragraph{Acknowledgements} I am grateful to Paula Belzig, Eric Culf, Sukanya Ghosal, Ray Laflamme, Debbie Leung, Graeme Smith, and Peixue Wu for helpful discussions; and Eric Culf for helpful comments on the manuscript. I am grateful to Anne Broadbent for invaluable support. 

Research at Perimeter Institute is supported in part by the Government of Canada through the Department of Innovation, Science, and Economic Development Canada and by the Province of Ontario through the Ministry of Colleges and Universities. The author was affiliated to Institute for Quantum Computing, University of Waterloo for part of this work, and acknowledges Norbert Lütkenhaus for funding and discussions during that time.

\bibliographystyle{bibtex/bst/alphaarxiv.bst}
\bibliography{bibtex/bib/full.bib,bibtex/bib/quantum.bib,bibtex/quantum_new.bib}

\newcommand{\etalchar}[1]{$^{#1}$}
\makeatletter\@ifundefined{url}{\newcommand{\url}[1]{\texttt{#1}}}{}\@ifundefined{href}{\newcommand{\href}[2]{\texttt{#2}}}{}\@ifundefined{mathbb}{\newcommand{\mathbb}[1]{#1}}{}\makeatother
\begin{thebibliography}{MHRW14}

\bibitem[AG89]{Arratia1989TutorialOL}
R.~Arratia and L.~Gordon.
\newblock Tutorial on large deviations for the binomial distribution.
\newblock {\em Bulletin of Mathematical Biology}, 51: 125--131, 1989.
\newblock \\ Online:
  \url{https://link.springer.com/article/10.1007/BF02458840}.

\bibitem[BFK09]{broadbent2009universal}
A.~Broadbent, J.~Fitzsimons, and E.~Kashefi.
\newblock Universal blind quantum computation.
\newblock In {\em 2009 50th annual IEEE symposium on foundations of computer
  science}, pages 517--526. IEEE, 2009.

\bibitem[Bla06]{Bla06}
B.~Blackadar.
\newblock {\em Operator algebras theory of {C*}-algebras and von {N}eumann
  algebras}.
\newblock Encyclopaedia of mathematical sciences, v. 122. Operator algebras and
  non-commutative geometry ; 3. Springer, Berlin, 2006.

\bibitem[BNS98]{Bar98}
H.~Barnum, M.~A. Nielsen, and B.~Schumacher.
\newblock Information transmission through a noisy quantum channel.
\newblock {\em Phys. Rev. A}, 57: 4153--4175, 1998.
\newblock \\
  \texttt{DOI:\,\href{http://dx.doi.org/10.1103/PhysRevA.57.4153}{10.1103/PhysRevA.57.4153}}.

\bibitem[CEHM99]{cirac1999distributed}
J.~I. Cirac, A.~Ekert, S.~F. Huelga, and C.~Macchiavello.
\newblock Distributed quantum computation over noisy channels.
\newblock {\em Physical Review A}, 59(6): 4249, 1999.

\bibitem[CLY97]{chuang1997bosonic}
I.~L. Chuang, D.~W. Leung, and Y.~Yamamoto.
\newblock Bosonic quantum codes for amplitude damping.
\newblock {\em Physical Review A}, 56(2): 1114, 1997.

\bibitem[CRS08]{cubitt2008structure}
T.~S. Cubitt, M.~B. Ruskai, and G.~Smith.
\newblock The structure of degradable quantum channels.
\newblock {\em Journal of Mathematical Physics}, 49(10), 2008.

\bibitem[CT05]{Cover:2005lom}
T.~M. Cover and J.~A. Thomas.
\newblock {\em {Elements of Information Theory}}.
\newblock Wiley, 2005.
\newblock \\
  \texttt{DOI:\,\href{http://dx.doi.org/10.1002/047174882x}{10.1002/047174882x}}.

\bibitem[Dev05]{Devetak}
I.~Devetak.
\newblock The private classical capacity and quantum capacity of a quantum
  channel.
\newblock {\em IEEE Transactions on Information Theory}, 51(1): 44--55, 2005.
\newblock \\
  \texttt{DOI:\,\href{http://dx.doi.org/10.1109/TIT.2004.839515}{10.1109/TIT.2004.839515}}.

\bibitem[DW05]{devetak2005distillation}
I.~Devetak and A.~Winter.
\newblock Distillation of secret key and entanglement from quantum states.
\newblock {\em Proceedings of the Royal Society A: Mathematical, Physical and
  engineering sciences}, 461(2053): 207--235, 2005.

\bibitem[FRT24]{fawzi2024capacities}
O.~Fawzi, M.~Rahaman, and M.~Taheri.
\newblock Capacities of quantum markovian noise for large times.
\newblock {\em arXiv preprint arXiv:2408.00116}, 2024.

\bibitem[GFY18]{guan2018decomposition}
J.~Guan, Y.~Feng, and M.~Ying.
\newblock Decomposition of quantum markov chains and its applications.
\newblock {\em Journal of Computer and System Sciences}, 95: 55--68, 2018.

\bibitem[Gro97]{grover1997quantum}
L.~K. Grover.
\newblock Quantum telecomputation.
\newblock {\em arXiv preprint quant-ph/9704012}, 1997.

\bibitem[Kim08]{kimble2008quantum}
H.~J. Kimble.
\newblock The quantum internet.
\newblock {\em Nature}, 453(7198): 1023--1030, 2008.

\bibitem[KL96]{knill1996theory}
E.~Knill and R.~Laflamme.
\newblock A theory of quantum error-correcting codes.
\newblock {\em arXiv preprint quant-ph/9604034}, 1996.

\bibitem[LCY97]{leung1997approximate}
D.~W. Leung, I.~L. Chuang, and Y.~Yamamoto.
\newblock Approximate quantum error correction can lead to better codes.
\newblock {\em Physical Review A}, 56(4): 2567, 1997.

\bibitem[Llo97]{Lloyd}
S.~Lloyd.
\newblock Capacity of the noisy quantum channel.
\newblock {\em Phys. Rev. A}, 55: 1613--1622, 1997.
\newblock \\
  \texttt{DOI:\,\href{http://dx.doi.org/10.1103/PhysRevA.55.1613}{10.1103/PhysRevA.55.1613}}.

\bibitem[LS09]{leung2009continuity}
D.~Leung and G.~Smith.
\newblock Continuity of quantum channel capacities.
\newblock {\em Communications in Mathematical Physics}, 292(1): 201--215, 2009.

\bibitem[LW23]{lami2023exact}
L.~Lami and M.~M. Wilde.
\newblock Exact solution for the quantum and private capacities of bosonic
  dephasing channels.
\newblock {\em Nature Photon.}, 17(arXiv: 2205.05736): 525--530, 2023.

\bibitem[MHRW14]{muller2014quantum}
A.~M{\"u}ller-Hermes, D.~Reeb, and M.~M. Wolf.
\newblock Quantum subdivision capacities and continuous-time quantum coding.
\newblock {\em IEEE Transactions on Information Theory}, 61(1): 565--581, 2014.

\bibitem[MSGL24]{mele2024quantum}
F.~A. Mele, F.~Salek, V.~Giovannetti, and L.~Lami.
\newblock Quantum communication on the bosonic loss-dephasing channel.
\newblock {\em Physical Review A}, 110(1): 012460, 2024.

\bibitem[RSW02]{ruskai2002analysis}
M.~B. Ruskai, S.~Szarek, and E.~Werner.
\newblock An analysis of completely-positive trace-preserving maps on
  \uppercase{M}$_2$.
\newblock {\em Linear algebra and its applications}, 347(1-3): 159--187, 2002.

\bibitem[SD24]{singh2024information}
S.~Singh and N.~Datta.
\newblock Information transmission under markovian noise.
\newblock {\em arXiv preprint arXiv:2409.17743}, 2024.

\bibitem[SD25]{singh2025capacities}
S.~Singh and N.~Datta.
\newblock Capacities of highly markovian divisible quantum channels.
\newblock {\em arXiv preprint arXiv:2504.10436}, 2025.

\bibitem[Shi17]{shirokov2017tight}
M.~Shirokov.
\newblock Tight uniform continuity bounds for the quantum conditional mutual
  information, for the holevo quantity, and for capacities of quantum channels.
\newblock {\em Journal of Mathematical Physics}, 58(10), 2017.

\bibitem[Sho02]{Shor}
P.~Shor.
\newblock The quantum channel capacity and coherent information.
\newblock {\em lecture notes, MSRI Workshop on Quantum Computation}, 2002.

\bibitem[SRD24]{singh2024zero}
S.~Singh, M.~Rahaman, and N.~Datta.
\newblock Zero-error communication under discrete-time markovian dynamics.
\newblock {\em arXiv preprint arXiv:2402.18703}, 2024.

\bibitem[SRW15]{szehr2015spectral}
O.~Szehr, D.~Reeb, and M.~M. Wolf.
\newblock Spectral convergence bounds for classical and quantum markov
  processes.
\newblock {\em Communications in Mathematical Physics}, 333: 565--595, 2015.

\bibitem[SY08]{smith2008quantum}
G.~Smith and J.~Yard.
\newblock Quantum communication with zero-capacity channels.
\newblock {\em Science}, 321(5897): 1812--1815, 2008.

\bibitem[VW23]{Vidick_Wehner_2023}
T.~Vidick and S.~Wehner.
\newblock {\em Introduction to Quantum Cryptography}.
\newblock Cambridge University Press, 2023.

\bibitem[WAR{\etalchar{+}}23]{wo2023resource}
K.~J. Wo, G.~Avis, F.~Rozp{\k{e}}dek, M.~F. Mor-Ruiz, G.~Pieplow,
  T.~Schr{\"o}der, L.~Jiang, A.~S. S{\o}rensen, and J.~Borregaard.
\newblock Resource-efficient fault-tolerant one-way quantum repeater with code
  concatenation.
\newblock {\em npj Quantum Information}, 9(1): 123, 2023.

\bibitem[Wat18]{watrous2018theory}
J.~Watrous.
\newblock {\em The theory of quantum information}.
\newblock Cambridge university press, 2018.

\end{thebibliography}

\end{document}